\documentclass{article}
\usepackage{amssymb}
\usepackage{amsmath}
\usepackage{amsthm}
\usepackage{graphicx}
\usepackage{placeins}
\usepackage[letterpaper,margin=1in]{geometry}
\usepackage[round]{natbib}
\usepackage{xcolor}
\definecolor{linkblue}{RGB}{31,58,140}
\usepackage[colorlinks=true,linkcolor=linkblue,citecolor=black,urlcolor=linkblue]{hyperref}

\hypersetup{
  pdftitle={Sufficiency Theorem and Closed-Form Posterior for a Windowed Change Point Detection Task},
  pdfauthor={Payton Hanks, Rishidev Chaudhuri, Timothy D. Hanks}
}

\newtheorem{theorem}{Theorem}

\title{Sufficiency Theorem and Closed-Form Posterior for a Windowed Change Point Detection Task}

\author{Payton Hanks$^{1}$
  \qquad Rishidev Chaudhuri$^{2,3,4,*}$
  \qquad Timothy D. Hanks$^{4,5,*}$\\[6pt]
  \small $^{1}$Davis Senior High School, Davis, CA 95616, USA\\
  \small $^{2}$Department of Neurobiology, Physiology and Behavior, University of California, Davis, Davis, CA 95616, USA\\
  \small $^{3}$Department of Mathematics, University of California, Davis, Davis, CA 95616, USA\\
  \small $^{4}$Center for Neuroscience, University of California, Davis, Davis, CA 95618, USA\\
  \small $^{5}$Department of Neurology, University of California, Davis, Sacramento, CA 95817, USA\\
  \small $^{*}$Corresponding authors: \texttt{rchaudhuri@ucdavis.edu}, \texttt{thanks@ucdavis.edu}}
\date{}

\begin{document}
\maketitle

\begin{abstract}
To function effectively in a changing world, agents must detect changes and act on them in a timely fashion. We study this in the context of a change point detection task in which clicks arrive from a Poisson process whose rate changes at an unknown time. An agent must declare the change within a detection window of fixed width or the trial ends in a miss. Here we prove a window sufficiency theorem, which states that the posterior probability of a change depends only on the clicks within the most recent window, regardless of the prior on the change time. This result follows from combining survival conditioning, which rules out change times whose windows have already expired, with cancellations driven by independent increments of a known baseline process. Window sufficiency generalizes to an unknown post-change process, as long as the baseline process is known. For an exponential prior on the change time, we further derive the posterior in closed form as a function of the elapsed times since each click in the most recent window, the known rates, and the window width. These results suggest that discarding old evidence can be optimal rather than a limitation when decisions face a deadline.
\end{abstract}

\section{Introduction}
Many decisions in perception and behavior need to be made against a deadline \citep{wald1945,luce1986,gold2007,hanks2017}. For example, a deer wandering in a forest must decide whether to run away when it senses a nearby pack of wolves in order to avoid being eaten. In a completely different scenario, a goalkeeper in soccer must decide whether or not to dive for the ball in order to stop it from going in the net. These scenarios share a common structure: a change of state in the environment occurs at an unknown time (e.g.\ the presence of wolves or the threat of a shot), and an agent has a bounded window to detect it and take action before the chance to respond is lost. Evidence of a change from before this window is no longer actionable: had the change occurred that long ago, the termination event (e.g.\ the deer being eaten, the goal being scored) would already have happened. An agent that conditions its decision on still being able to act can therefore disregard such evidence, provided it knows what the environment looks like when nothing has changed.

It is well established that uniform weighting of evidence in a changing environment is suboptimal. In ideal-observer models of such environments, the influence of past evidence is governed by the rate at which the environment changes \citep{wilson2010,glaze2015,velizcuba2016,radillo2017}, and agents adjust the timescale over which they accumulate evidence to match this \citep{glaze2015,piet2018}. Here we examine whether a response deadline imposes another reason for discarding old evidence.

We formalize this scenario with the following task (\autoref{fig:task}), modeled on an auditory change detection paradigm used in human psychophysics and decision neuroscience studies \citep{johnson2017,ganupuru2019,harun2020,booras2021,ganupuru2025,goldring2026}. Clicks arrive from a Poisson process with known baseline rate $\lambda_0$ and change to a known rate $\lambda_1$ at some unknown change point $\tau$. The change has a detection window of width $w$, so if a declaration is not made by $\tau + w$, the task terminates with a miss. We are interested in the agent's posterior belief that the change has already occurred, given the evidence observed so far. We treat $\lambda_1$ as known for most of the analysis but also consider an unknown post-change process in \autoref{sec:general}.

\begin{figure}[t]
\centering
\includegraphics[width=\textwidth]{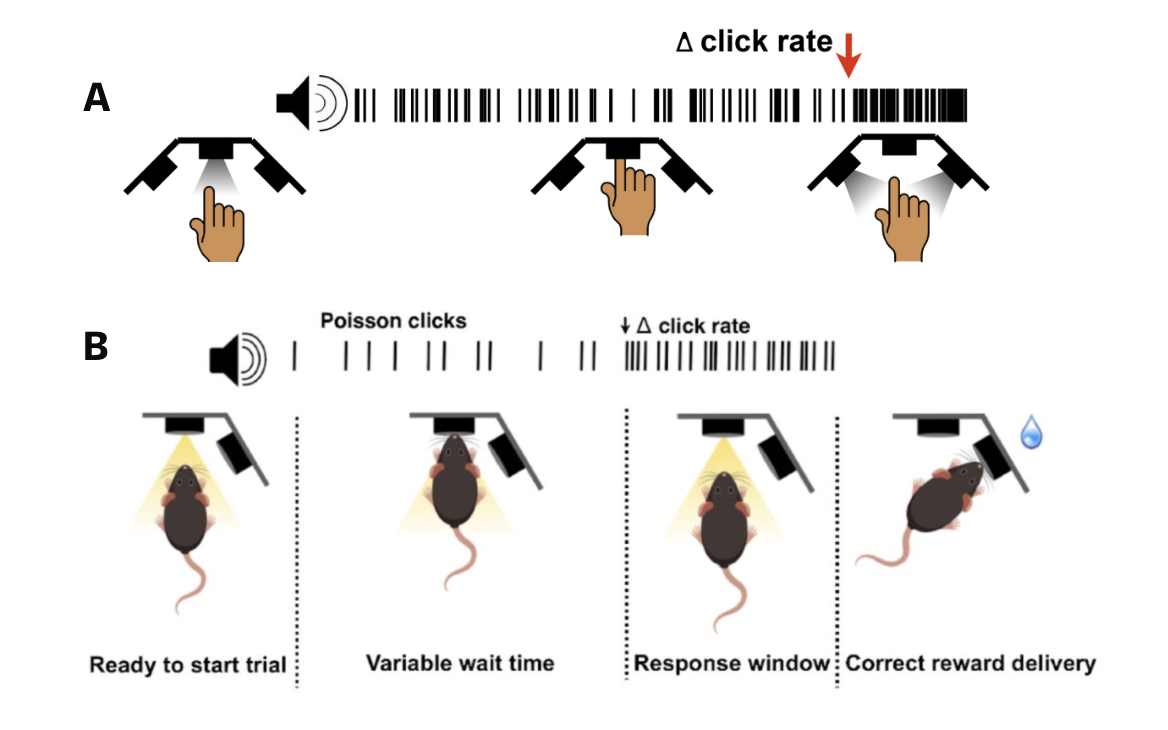}
\caption{Auditory change detection task. (A) Human version. Subjects hold a finger in a central port while listening to a stream of clicks generated by a Poisson process, and withdraw it to report a detected increase in click rate (red arrow). (B) Rat version. Rats start a trial at a central port and hold through a variable wait time while Poisson clicks play. After the change in click rate, they must respond within a fixed response window to receive a water reward. In both versions, a change that is not reported within the response window counts as a miss, the deadline structure modeled here with window width $w$. Panels adapted from previous studies \citep{harun2020, ganupuru2025}.}
\label{fig:task}
\end{figure}

We condition the posterior calculation on the task being active: a miss has not yet occurred and the agent has not made a declaration. This conditioning is natural from the perspective of the agent, who would know whether the task is still active. It is also important in itself: the event $\{\text{no miss by } t\}$ is informative, ruling out the possibility of change points prior to $t - w$, whereas conditioning on no declaration turns out to be inconsequential. This choice is not merely a technical convenience. It more accurately captures the intent of the task, representing the decision-making process of the agent rather than that of a detached observer.

In this paper we have three main contributions. First, we set up the active-conditioned posterior for a constant-hazard Poisson model. Second, we prove a window sufficiency theorem, showing that clicks prior to $t - w$ have no effect on the posterior, and that this holds as a consequence of the task's termination structure together with the known baseline, rather than of a particular prior. It follows that the optimal inference depends only on evidence within the most recent window, so a bounded-memory strategy that discards older evidence loses no information when the environment imposes a fixed response deadline. This offers a normative account of the finding that observers adopt shorter timescales of evidence evaluation when the signal to be detected is briefer \citep{ossmy2013,harun2020}. We further show that the result extends to an unknown post-change process, so that the agent need only know the baseline, not the change magnitude it is attempting to detect. Finally, we evaluate the posterior in closed form for an exponential prior on the change point (constant hazard $\mu$), computable directly from observed and known values without numerical integration. These results give the posterior at any single observation time, which can be recomputed from the latest window as clicks arrive. A dynamical rule that instead updates the posterior incrementally with each click, as has been derived for changing environments without a deadline \citep{glaze2015,radillo2017}, and the optimal declaration policy it enables, are left to future work.

\section{Posterior for Change Detection}
\label{sec:posterior-section}

\subsection{Setup and notation}
\label{sec:setup}
Clicks arrive as a Poisson process observed on $[0,t]$. The rate begins at $\lambda_0$ and switches to $\lambda_1$ at change point $\tau$ with constant hazard $\mu$, so $\tau$ is exponentially distributed with prior density
\[
\pi(s) = \mu e^{-\mu s}, \qquad s \ge 0.
\]
The detection window has width $w$: a change at $\tau$ must be declared by $\tau+w$. Otherwise, it counts as a miss. We set no fixed trial length, so the task terminates only at a declaration or a miss.

The quantities below depend on how long ago events occurred rather than on the current time itself, so we choose to parameterize a possible change by its age, $a = t-s$, for a candidate change at time $s$. For example, a change at observation time $t$ has age $a = 0$, while a change at the beginning of the window has age $a = w$. We write $D$ for the full set of click times observed on $[0,t]$, $D_I \subseteq D$ for the clicks whose ages fall in an interval $I=[a_1,a_2)$, and $n_I$ for the number of clicks in $I$.

At the observation time $t$, two hypotheses remain. Under the not-yet-changed hypothesis $S_0 = \{\tau > t\}$, the change has not yet occurred and the rate is $\lambda_0$ throughout $[0,t]$. Under the alternative $S_1 = \{\tau \le t\}$, the change has occurred at some age $a = t-\tau \in [0,t]$, so clicks younger than $a$ arrive at rate $\lambda_1$ and clicks older than $a$ at rate $\lambda_0$. It is useful to also name a finer family of hypotheses for any fixed age of the change,
\[
H_a = \{\tau = t-a\}, \qquad a \in [0,t],
\]
each of which is a special case of the previously defined hypothesis:
\[
S_1 = \bigcup_{a \in [0,t]} H_a, \qquad H_a \subseteq S_1 \text{ for every } a \in [0,t].
\]
Because $\tau$ is continuous, each individual $H_a$ has probability zero, and its prior and posterior weights are densities in $a$. We write $P(H_a) = \pi(t-a)$ for the prior density and $P(H_a \mid D)$ for the posterior density, so that $P(S_1 \mid D) = \int_{0}^{t} P(H_a \mid D)\,da$. Survival conditioning (\autoref{sec:survival}) will show that $H_a$ carries zero posterior weight for $a \ge w$, leaving only the most recent window to integrate over. It is worth stating that the joint event $H_a \cap S_1$ is just $H_a$, so
\[
P(H_a, S_1) = P(H_a) = \pi(t-a).
\]
This is simply a consequence of the inclusion $H_a \subseteq S_1$, stating that the probability that the change has already occurred and that it occurred $a$ time units ago is simply the probability that it occurred $a$ time units ago.

\subsection{Survival conditioning}
\label{sec:survival}
At the observation time $t$ we condition on the trial still being ``alive'': (i) no miss has occurred, and (ii) the observer has not yet declared a change. Since the declaration rule is a deterministic function of the observed clicks, the event $\{\text{no declaration by } t\}$ is measurable by $D$; it contributes the same factor to each hypothesis and cancels from the posterior, providing no additional information on $\tau$. The event $\{\text{no miss by } t\}$ is informative. Given no declaration, a change at $\tau$ is missed if $\tau+w<t$, that is, if its age $t-\tau$ exceeds $w$, so ``no miss'' is equivalent to $a<w$. After conditioning, the only reachable states are
\[
\{a \in [0,w)\}\ \cup\ S_0,
\]
i.e.\ either the change occurred within the last window (a subset of hypothesis $S_1$), or it has not yet occurred (hypothesis $S_0$). Equivalently, $P(H_a \mid D) = 0$ for $a \ge w$.

Conditioning on $\{\tau > t-w\}$ also rescales the prior weights of the surviving hypotheses. The conditional prior density of the age of the change on $[0,w)$ is $\pi(t-a)/P(\tau > t-w)$, and the conditional prior probability of $S_0$ is $P(S_0)/P(\tau > t-w)$. Since the same factor $1/P(\tau > t-w) = e^{\mu(t-w)}$ multiplies every surviving weight, it appears in both the numerator and the denominator of any posterior and cancels. We therefore work with the unconditional weights $\pi(t-a)$ and $P(S_0)$ throughout, with the understanding that they are proportional to, rather than equal to, the conditional prior. For $t<w$ no change hypothesis has yet expired. All expressions below then hold with $w$ replaced by $t$.

\subsection{Likelihood}
\label{sec:likelihood}
For a homogeneous Poisson process of rate $\lambda$ observed on an interval $I=[a_1,a_2)$ containing $n_I$ clicks,
\[
L(\lambda, D_I) = \lambda^{\,n_I}\, e^{-\lambda (a_2-a_1)}.
\]
Under $H_a$, independent increments give
\[
P(D \mid H_a) = L(\lambda_1, D_{[0,a)})\,L(\lambda_0, D_{[a,t]}),
\]
which splits the likelihood at the change point so that each factor carries a constant click rate. Clicks younger than $a$ arrive at rate $\lambda_1$ and clicks older than $a$ at rate $\lambda_0$. Under $S_0$,
\[
P(D \mid S_0) = L(\lambda_0, D_{[0,t]}).
\]

\subsection{Posterior}
\label{sec:posterior}
Throughout this and the following sections, every posterior probability is conditioned on the trial being active at $t$ (no miss and no declaration yet). We suppress this conditioning from the notation, writing $P(\,\cdot \mid D)$ for $P(\,\cdot \mid D,\,\text{active})$. Prior weights are unconditional, as justified at the end of \autoref{sec:survival}. Applying Bayes' theorem, the posterior probability of a change takes the generic form
\[
P(S_1 \mid D) = \int_{0}^{w} P(H_a \mid D)\,da
= \frac{\displaystyle\int_{0}^{w} \pi(t-a)\,P(D \mid H_a)\,da}
       {P(S_0)\,P(D \mid S_0) \;+\; \displaystyle\int_{0}^{w} \pi(t-a)\,P(D \mid H_a)\,da}.
\]
Because $H_a \subseteq S_1$, the weight attached to each candidate age is simply its prior density $\pi(t-a) = \mu e^{-\mu (t-a)}$ (\autoref{sec:setup}). The not-yet-changed hypothesis carries prior probability
\[
P(S_0) = P(\tau > t) = 1 - P(\tau \le t) = 1- \int_{0}^{t} \pi(s)\, ds = e^{-\mu t}.
\]
Substituting these together with the likelihood functions,
\[
P(S_1 \mid D)
= \frac{\displaystyle\int_{0}^{w}
        \mu e^{-\mu (t-a)}\,L(\lambda_1,D_{[0,a)})\,L(\lambda_0,D_{[a,t]})\,da}
       {e^{-\mu t}\,L(\lambda_0,D_{[0,t]})
        \;+\;\displaystyle\int_{0}^{w}
        \mu e^{-\mu (t-a)}\,L(\lambda_1,D_{[0,a)})\,L(\lambda_0,D_{[a,t]})\,da}.
\]
Note that this expression still depends on clicks older than $w$: the surviving likelihood factors $L(\lambda_0, D_{[a,t]})$ and $L(\lambda_0, D_{[0,t]})$ both span the ages $[w,t]$ before the window. The next section will show that this dependence is only apparent.

\section{Window Sufficiency Theorem}
\label{sec:sufficiency}

\subsection{Known Post-Change Rate}
\label{sec:poisson}

\begin{theorem}[Window sufficiency, Poisson case]
\label{thm:poisson}
Conditional on the trial being active at time $t$, the posterior density $P(H_a \mid D)$ for every $a \in [0,w)$, and hence the posterior probability $P(S_1 \mid D)$, depends on the data only through the clicks in the most recent window, $D_{[0,w)}$.
\end{theorem}

\begin{proof}
Survival conditioning has already restricted the age of the change to $[0,w)$, so it remains to show that the clicks older than $w$ carried by the likelihood factors cancel. For any candidate age $a \in [0,w)$, split the pre-change factor at age $w$:
\[
L(\lambda_0, D_{[a,t]})
= \lambda_0^{\,n_{[a,w)}}\,e^{-\lambda_0 (w-a)}
  \;\cdot\; \underbrace{\lambda_0^{\,n_{[w,t]}}\,e^{-\lambda_0 (t-w)}}_{\textstyle C}.
\]
We denote this common pre-window factor by $C$. Under the not-yet-changed hypothesis, we can make the same split,
\[
L(\lambda_0, D_{[0,t]})
= \lambda_0^{\,n_{[0,w)}}\,e^{-\lambda_0 w} \;\cdot\; C,
\]
because under $S_0$ the clicks older than $w$ also arrive at rate $\lambda_0$. The factor $C = \lambda_0^{\,n_{[w,t]}} e^{-\lambda_0 (t-w)}$ is thus common to the $S_0$ term and to every $H_a$ term, so it factors out of both the numerator and denominator and cancels. What remains involves only $n_{[0,a)}$ and $n_{[a,w)}$, i.e.\ clicks in the most recent window. Note that this result does not depend on the prior chosen for this task; it comes from the terminating structure of the response window together with the known baseline rate.
\end{proof}

After cancellation, the posterior density over the age of the change is
\[
P(H_a \mid D) \;\propto\;
   \mu e^{-\mu (t-a)}\,\lambda_1^{\,n_{[0,a)}}\,\lambda_0^{\,n_{[a,w)}}\,
   e^{-\lambda_1 a - \lambda_0 (w-a)},\qquad a\in[0,w),
\]
with the exponential prior $\mu e^{-\mu (t-a)}$ appearing explicitly (it is not constant in $a$). The posterior weight on the not-yet-changed hypothesis is
\[
P(S_0 \mid D) \;\propto\; e^{-\mu t}\,\lambda_0^{\,n_{[0,w)}} e^{-\lambda_0 w},
\]
with the same constant of proportionality. The likelihood ratio of a candidate age $H_a$ against the not-yet-changed hypothesis $S_0$ is
\[
\mathrm{LR}(a) = \frac{P(D \mid H_a)}{P(D \mid S_0)}
= \left(\frac{\lambda_1}{\lambda_0}\right)^{n_{[0,a)}}
   e^{-(\lambda_1-\lambda_0)\,a},
\]
depending on the data only through the clicks younger than $a$ and on the age $a$ itself. The hazard $\mu$ is absent because the likelihood ratio is exactly the factor by which the data update the prior odds:
\[
\underbrace{\frac{P(H_a \mid D)}{P(S_0 \mid D)}}_{\text{posterior odds}}
= \underbrace{\frac{\pi(t-a)}{P(S_0)}}_{\text{prior odds}}
  \;\times\; \mathrm{LR}(a)
= \frac{\mu e^{-\mu (t-a)}}{e^{-\mu t}}\;\mathrm{LR}(a)
= \mu e^{\mu a}\;\mathrm{LR}(a),
\]
so all dependence on $\mu$ is confined to the prior-odds factor, and the likelihood ratio represents the evidence the clicks alone provide for a change of age $a$. For the exponential prior, the observation time $t$ has also dropped out, reflecting the memorylessness of the exponential prior. For a general prior, the prior odds $\pi(t-a)/P(\tau>t)$ still depend on $t$, while the likelihood ratio does not. The posterior odds of a change against no change follow by integrating over the window,
\[
\frac{P(S_1 \mid D)}{P(S_0 \mid D)} = \int_{0}^{w} \frac{\pi(t-a)}{P(S_0)}\,\mathrm{LR}(a)\,da .
\]

The proof uses two components, and it is worth isolating them because \autoref{sec:general} will show that they are all that is needed. First, termination on a miss restricts the surviving ages of the change to $a < w$, so that under every surviving hypothesis all clicks older than $w$ precede the change. Second, those clicks come from the known baseline process, so their likelihood is the same known number $C$ under every surviving hypothesis. Neither requires knowledge of the post-change rate $\lambda_1$. \autoref{fig:schematic} illustrates both components.

\begin{figure}[t]
\centering
\includegraphics[width=\textwidth]{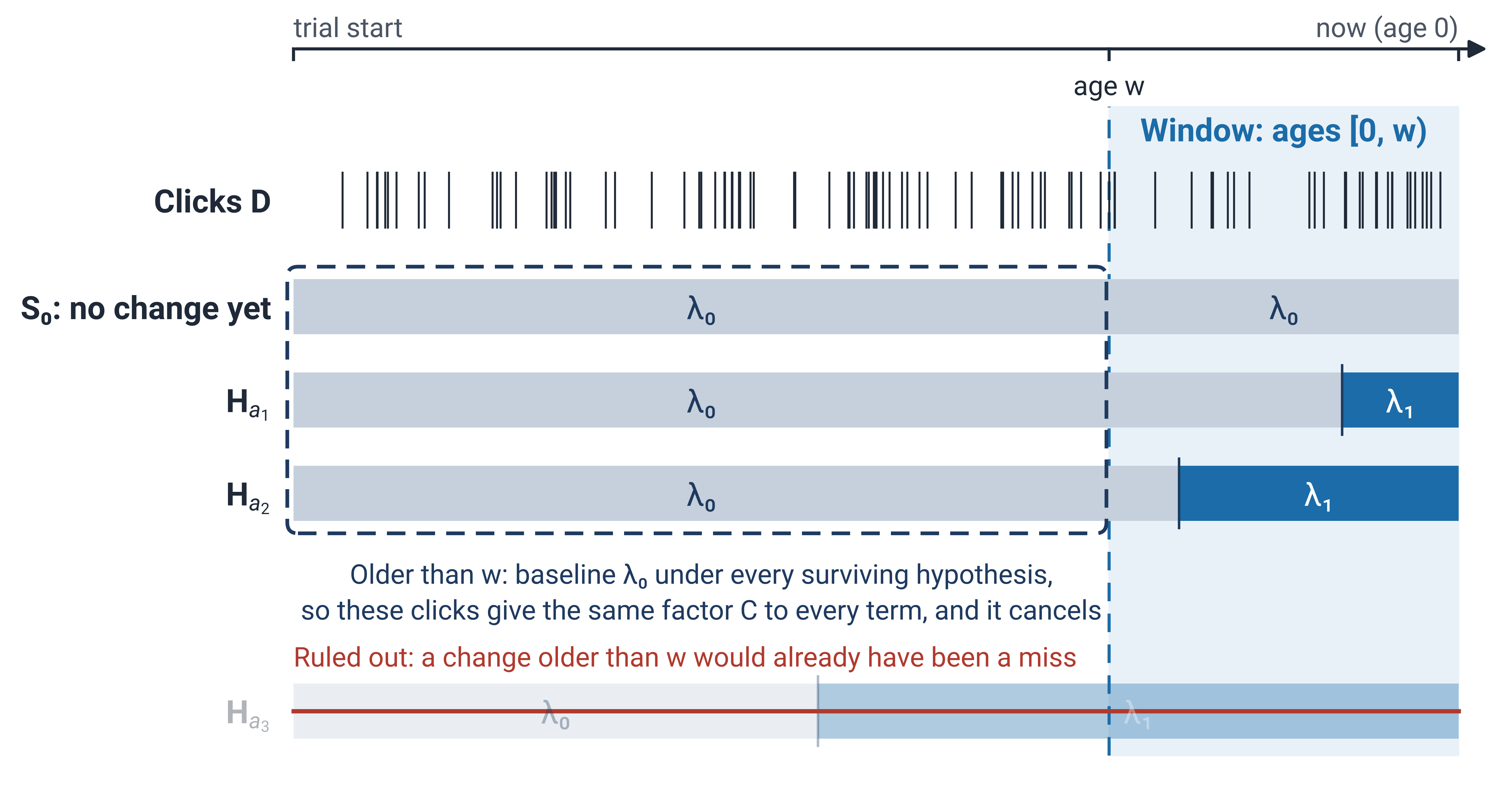}
\caption{Window sufficiency. A single stream of clicks $D$ (top) is evaluated under the surviving hypotheses: no change yet, $S_0$, and two example changes, $H_{a_1}$ and $H_{a_2}$, whose ages $a_1$ and $a_2$ are both less than $w$. Each bar shows the click rate a hypothesis assigns at each age, $\lambda_0$ before the change and $\lambda_1$ after it. Survival conditioning rules out changes older than $w$ (e.g., $H_{a_3}$), which would already have ended the trial in a miss. Every surviving hypothesis therefore assigns the baseline rate $\lambda_0$ to all clicks older than $w$ (dashed box), so these clicks contribute the same factor $C$ to every term of Bayes' theorem and cancel from the posterior.}
\label{fig:schematic}
\end{figure}

\subsection{Generalization: Departure from Known Baseline}
\label{sec:general}

The two components isolated above make no reference to the post-change process beyond the fact that it acts only on the clicks after the change. This suggests that window sufficiency should survive when the agent does not know what it is looking for, only what the world looks like when nothing has happened. We now make this precise.

Let $\theta$ denote the (possibly unknown) description of the post-change process, drawn from a prior $\pi_\theta$ independent of $\tau$. While the basic model assumes the post-change behavior ($\theta = \lambda_1$) is a single known number, $\theta$ can be generalized to fit more complex scenarios. For example, $\theta$ can represent an unknown rate that must be estimated, a vector of multiple values (such as a new mean and variance), or a discrete index representing different types of post-change processes. Correspondingly, we replace the Poisson likelihoods $L(\lambda_0, \cdot)$ and $L(\lambda_1, \cdot)$ with general likelihoods $L_0(\cdot)$ and $L_1(\theta, \cdot)$ for the baseline and post-change processes. We make the following assumptions.

\begin{enumerate}
\item[(A1)] \emph{Termination.} The trial ends in a miss if no declaration is made by $\tau + w$.
\item[(A2)] \emph{Known baseline.} Before the change, clicks follow a process whose law is fully known to the agent. We write $L_0(D_I)$ for its likelihood on an interval $I$.
\item[(A3)] \emph{Independent increments.} Clicks on disjoint intervals are statistically independent, so the baseline likelihood factorizes for $a_1<a_2<a_3$,
\[
L_0(D_{[a_1,a_3)}) = L_0(D_{[a_1,a_2)})\,L_0(D_{[a_2,a_3)}).
\]
\item[(A4)] \emph{Post-change process acts only after the change.} Given a change of age $a$ and $\theta$, the clicks younger than $a$ have likelihood $L_1(\theta,\, D_{[0,a)})$ and are independent of the clicks older than $a$.
\end{enumerate}

These four assumptions restate, in general form, what the Poisson model of \autoref{sec:setup} already satisfies and what the proof of \autoref{thm:poisson} already used. The Poisson model uses $L_0(D_I) = L(\lambda_0, D_I)$ in (A2)--(A3) and $L_1(\lambda_1, D_{[0,a)}) = L(\lambda_1, D_{[0,a)})$ in (A4). The assumptions also line up with the two components isolated at the end of \autoref{sec:poisson}. The first component, that the termination structure of the task restricts the surviving ages of the change, is (A1). The second, that the likelihood of the clicks older than the window is the same known quantity under every surviving hypothesis, is (A2) and (A3) together. Assumption (A4) is the general form of the factorization used in \autoref{sec:likelihood} to write $P(D \mid H_a)$, which is part of the model's setup rather than the mechanism behind window sufficiency.

\begin{theorem}[Window sufficiency, general form]
\label{thm:general}
Under (A1)--(A4), and conditional on the trial being active at time $t$, the posterior density $P(H_a \mid D)$ for every $a \in [0,w)$ and the posterior probability $P(S_1 \mid D)$ depend on the data only through $D_{[0,w)}$, for any prior $\pi$ on the change time and any prior $\pi_\theta$ on the post-change process.
\end{theorem}

\begin{proof}
By (A1) and survival conditioning, the surviving hypotheses are $S_0$ and $H_a$ for $a \in [0,w)$, exactly as before. Fix $a$ and $\theta$. By (A4) and then (A3),
\[
P(D \mid H_a, \theta) = L_1(\theta, D_{[0,a)})\,L_0(D_{[a,t]})
= L_1(\theta, D_{[0,a)})\,L_0(D_{[a,w)})\;\underbrace{L_0(D_{[w,t]})}_{\textstyle C},
\]
and likewise $P(D \mid S_0) = L_0(D_{[0,w)})\,C$. Since $C$ does not depend on $\theta$ by (A2), integrating over the prior on $\theta$ gives
\[
P(D \mid H_a) = \int P(D \mid H_a,\theta)\,\pi_\theta(d\theta)
= C\;L_0(D_{[a,w)})\;\underbrace{\int L_1(\theta, D_{[0,a)})\,\pi_\theta(d\theta)}_{\textstyle \bar L_1(D_{[0,a)})}.
\]
By Bayes' theorem,
\[
P(H_a \mid D) = \frac{\pi(t-a)\,P(D \mid H_a)}
  {P(S_0)\,P(D \mid S_0) + \displaystyle\int_{0}^{w} \pi(t-b)\,P(D \mid H_{b})\,db}.
\]
Every likelihood on the right carries the factor $C$, so it cancels, leaving
\[
P(H_a \mid D) = \frac{\pi(t-a)\,L_0(D_{[a,w)})\,\bar L_1(D_{[0,a)})}
  {P(S_0)\,L_0(D_{[0,w)}) + \displaystyle\int_{0}^{w} \pi(t-b)\,L_0(D_{[b,w)})\,\bar L_1(D_{[0,b)})\,db},
\]
which depends on $D$ only through $D_{[0,w)}$. Integrating over $a \in [0,w)$ gives the same for $P(S_1 \mid D)$.
\end{proof}

The theorem says that under a response deadline the agent may discard pre-window evidence not only when it knows which change to expect, but whenever it knows the baseline. The deer in the introduction does not need to know the threat it is scanning for, but rather what the forest sounds like when it is safe. Under every surviving hypothesis, each pre-window click is a sample from that known safe state, so it provides the deer with no additional information. The likelihood ratio for a candidate age becomes
\[
\mathrm{LR}(a) = \frac{\bar L_1(D_{[0,a)})}{L_0(D_{[0,a)})},
\]
which compares how well the clicks younger than $a$ are explained by the (prior-averaged) alternative versus the baseline, and the odds formulas above carry over with this $\mathrm{LR}(a)$ in place of the Poisson one.

Of the four assumptions, (A1) and (A4) define the task itself, while (A2) and (A3) are the substantive restrictions on the baseline, and the theorem fails in a different way without each. If the baseline rate $\lambda_0$ were unknown, violating (A2), pre-window clicks would become informative about it and thus about the in-window likelihoods under every hypothesis for the baseline rate. The factor $C$ would then depend on the hypothesis-specific posterior over $\lambda_0$ and would not cancel. If instead the baseline had memory, violating (A3), the likelihood of the in-window clicks would depend on pre-window clicks through the conditioning, even with all parameters known. Window sufficiency is a statement about detecting a departure from a fully known, memoryless baseline.

\section{Closed-Form Evaluation}
\label{sec:closedform}
We return to the Poisson model with known rates. Window sufficiency reduces the posterior to a function of the clicks in the most recent window alone. The counts $n_{[0,a)}$ and $n_{[a,w)}$ are piecewise constant in $a$, only changing as the candidate age $a$ crosses the age of an observed click. Thus, the integral evaluates in closed form on intervals between clicks.

Label the ages of the clicks in the window $[0,w)$, from youngest to oldest, as
\[
0 \le \alpha_1 < \alpha_2 < \cdots < \alpha_m < w,\qquad m = n_{[0,w)},
\]
and set the boundary points $\alpha_0 = 0$ and $\alpha_{m+1} = w$. For $a$ in the open sub-interval $(\alpha_j,\alpha_{j+1})$ exactly $j$ clicks are younger than $a$ and $m-j$ are older, so $n_{[0,a)} = j$ and $n_{[a,w)} = m-j$ are constant there. The count weights $\lambda_1^{\,n_{[0,a)}}\lambda_0^{\,n_{[a,w)}}$ are piecewise constant in $a$ for any prior, so the decomposition into between-click sub-intervals is a consequence of sufficiency alone. Writing
\[
\psi(a) = e^{-\lambda_1 a - \lambda_0 (w-a)}
\]
for the likelihood's exponential factor, the change-hypothesis weight under a general change-point prior $\pi$, with the common pre-window factor $C$ already canceled, is
\[
Z_1 = \int_{0}^{w} \pi(t-a)\,\lambda_1^{\,n_{[0,a)}}\lambda_0^{\,n_{[a,w)}}\,
      \psi(a)\,da
    = \sum_{j=0}^{m}\lambda_1^{\,j}\lambda_0^{\,m-j}
      \int_{\alpha_j}^{\alpha_{j+1}} \pi(t-a)\,\psi(a)\,da.
\]
Here the count weights are exact and prior-independent. Only the integrals $\int \pi\,\psi$ depend on the prior. A closed form results whenever $\pi(t-a)\psi(a)$ has an elementary antiderivative, e.g.\ a uniform prior ($\pi$ constant, giving denominator $\lambda_0-\lambda_1$) or the exponential prior $\pi(s)=\mu e^{-\mu s}$ adopted below. For a general $\pi$, the computation reduces to $m+1$ one-dimensional integrals over between-click intervals.

For our exponential prior (\autoref{sec:setup}), $\pi(t-a) = \mu e^{-\mu t} e^{\mu a}$. Write
\[
\phi(a) = e^{\mu a}\,\psi(a) = e^{\mu a -\lambda_1 a - \lambda_0 (w-a)},
\]
so that $\pi(t-a)\psi(a) = \mu e^{-\mu t}\,\phi(a)$ and the change-hypothesis weight is
\[
Z_1 = \mu e^{-\mu t} \int_{0}^{w}
   \lambda_1^{\,n_{[0,a)}}\lambda_0^{\,n_{[a,w)}}\,\phi(a)\,da
= \mu e^{-\mu t} \sum_{j=0}^{m}\lambda_1^{\,j}\lambda_0^{\,m-j}
   \int_{\alpha_j}^{\alpha_{j+1}}\phi(a)\,da.
\]
Since $\phi'(a) = (\lambda_0+\mu-\lambda_1)\,\phi(a)$, assuming $\lambda_1 - \lambda_0 \neq \mu$, each piece integrates to
\[
\int_{\alpha_j}^{\alpha_{j+1}}\phi(a)\,da
= \frac{\phi(\alpha_{j+1}) - \phi(\alpha_j)}{\lambda_0+\mu-\lambda_1},
\]
so that
\[
Z_1 = \frac{\mu e^{-\mu t}}{\lambda_0+\mu-\lambda_1}
   \sum_{j=0}^{m}\lambda_1^{\,j}\lambda_0^{\,m-j}
   \bigl[\phi(\alpha_{j+1}) - \phi(\alpha_j)\bigr].
\]
If instead $\lambda_1 - \lambda_0 = \mu$, then $\phi$ is constant and each piece is simply $\phi(\alpha_j)\,(\alpha_{j+1}-\alpha_j)$. The not-yet-changed hypothesis contributes the single term
\[
Z_0 = \underbrace{e^{-\mu t}}_{P(S_0)}\,\lambda_0^{\,m}\,e^{-\lambda_0 w},
\]
and the posterior probability of a change is
\[
P(S_1 \mid D) = \frac{Z_1}{Z_0 + Z_1}.
\]
The factor $e^{-\mu t}$ is common to $Z_0$ and $Z_1$ and cancels, so the posterior odds take the form
\[
\frac{P(S_1 \mid D)}{P(S_0 \mid D)} = \frac{Z_1}{Z_0}
= \frac{\mu}{\lambda_0+\mu-\lambda_1}
  \sum_{j=0}^{m}\left(\frac{\lambda_1}{\lambda_0}\right)^{j}
  \Bigl[e^{(\lambda_0+\mu-\lambda_1)\,\alpha_{j+1}} - e^{(\lambda_0+\mu-\lambda_1)\,\alpha_j}\Bigr].
\]

For the exponential prior, the posterior depends on the data only through the ages $\alpha_1,\dots,\alpha_m$ of the clicks in the window, each found by subtracting the click time from the observation time $t$, together with the fixed model parameters $\lambda_0,\lambda_1$, $\mu$, and the window width $w$. The observation time $t$ itself does not appear. The posterior is computed in closed form with no numerical integration.

\autoref{fig:trial} applies the closed form to a simulated trial. For comparison, the figure shows the posterior of the same model without a response deadline, in which every change time since the start of the trial remains possible, as well as the posterior density over the age of the change, which is confined to the most recent window.

\begin{figure}[!htbp]
\centering
\includegraphics[width=0.75\textwidth]{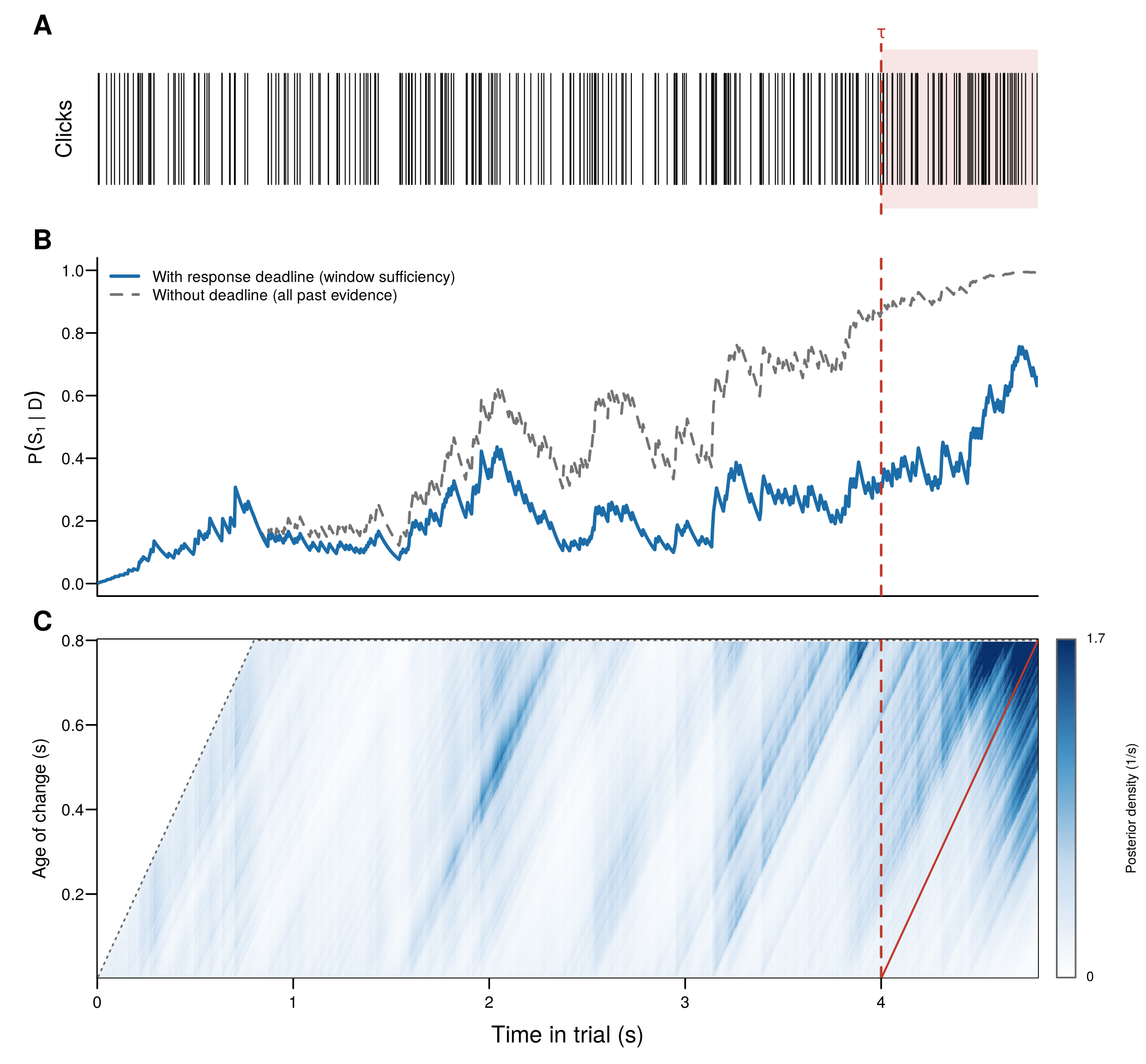}
\caption{Closed-form posterior on a simulated trial ($\lambda_0 = 50$~Hz, $\lambda_1 = 60$~Hz, $w = 0.8$~s, $\mu = 0.25$~s$^{-1}$, change at $\tau = 4$~s; the trial ends in a miss at $\tau + w$). (A) Click times. The dashed line marks the change, and the shading signifies the post-change period. (B) Posterior probability of a change, $P(S_1 \mid D)$, computed from the clicks in the most recent window (\autoref{sec:closedform}; solid), and for the same model without a response deadline, in which all change times since the start of the trial remain possible (dashed). (C) Posterior density over the age of the change, $P(H_a \mid D)$, on a fixed color scale. Candidate changes enter at age $0$ and expire at age $w$. For $t < w$ the density is supported on $[0,t)$ (dotted line). The red line marks the age of the true change.}
\label{fig:trial}
\end{figure}

\FloatBarrier
\section{Discussion}
\label{sec:discussion}

For an agent observing a Poisson process under a response deadline, we have shown that evidence relevant to the change point is confined to the most recent window, i.e.\ events prior to the window have no effect on the change probability, for any prior on the change point. Conditioning on the task being active removes the change hypotheses whose windows have already expired, and independent increments of the known baseline process remove the contribution of events preceding the window entirely. Neither step uses any property of the post-change process, so the result extends to the case where the agent only knows the baseline (\autoref{thm:general}). For the known-rate model, we derived a likelihood ratio for any candidate change and a closed-form posterior for an exponential distribution of change points. We stress that this closed-form posterior concerns only the inference of whether a change occurred, not the optimality of any decision itself potentially based on the posterior.

The structure of this task recurs whenever a fleeting change must be detected and acted on within a fixed window. Consider the earlier example of a deer, who needs to be aware of evidence that signals an attack from a pack of wolves so it can flee. A cue that the wolves were about to attack long ago is uninformative about whether the deer should flee. Had the attack begun far enough in the past, its window would have closed and the trial would have ended in a miss: in this case, the deer being eaten. A deer conditioning on the fact that it is still in a position to act should, optimally, disregard evidence older than a certain response window (real predator evasion is of course more complex; the point is structural). The generalization sharpens this picture, showing that the deer does not need to scan for wolves in particular. As long as it knows what the forest is like when it is safe, it may discard old evidence while remaining alert to any departure from that baseline, whether the threat turns out to be wolves, a hunter, or a fire (however, different threats may have distinct response windows). This shows forgetting of evidence can be an optimal feature rather than a limitation in a behavioral model where there is a response deadline. The same logic applies to the soccer goalkeeper, who should, optimally, disregard older evidence for a shot in deciding whether to dive. In each case, the imposed deadline is exactly what makes recent evidence sufficient, and the ``forgetting'' of older evidence reflects the structure of the task rather than imperfection in the agent.

Our result adds to previous work showing that discarding older evidence can be optimal. When an environment is subject to frequent change, an ideal observer weighs past evidence less as it ages \citep{wilson2010,glaze2015,velizcuba2016,radillo2017} and matches how quickly they discount evidence to the volatility of their environment \citep{glaze2015,piet2018}. In both these changing environments and our task, how much weight evidence receives depends on whether it is still relevant to a decision about a potential change, and older evidence tends to be less relevant than recent evidence. However, the way evidence becomes irrelevant is fundamentally different between our task and past work. Without a response deadline, evidence is discounted according to the probability that the environment has changed since the evidence occurred. This leads to a smooth decay in the influence of past evidence, frequently captured by leaky integration \citep{ossmy2013,glaze2015}, which fits naturally within frameworks of decision making based on evidence accumulation \citep{ratcliff1998,usher2001,gold2007,ratcliff2008,bogacz2010,forstmann2016,schall2019}. In our setting, evidence is discarded because a change would impose a deadline to respond: had the change occurred before the most recent window, its deadline would already have passed, so evidence from before the window provides no information on any change that is still possible (as proven with window sufficiency). This leads to a sharp boundary at the edge of the response window, beyond which evidence has no influence at all. This boundary may be softened in practice when the response window itself is uncertain (e.g.\ different threats allow different amounts of time to respond), since evidence of a given age then remains relevant only to the extent that the window could still be open. Similar mechanisms might then serve both settings, a possibility we return to below in discussing the online dynamics of the posterior.

Importantly, several assumptions we make delimit the scope of our claims. The sufficiency result requires a fully known, memoryless baseline. If the baseline rate $\lambda_0$ were unknown, clicks preceding the window would be informative about it and thus re-enter the inference, and a baseline with memory would let pre-window clicks shape the in-window likelihoods even with all other parameters known. Knowledge of the post-change rate $\lambda_1$ is not required for sufficiency, although the closed-form posterior of \autoref{sec:closedform} relies on it. We also assume a single, non-reverting change, in contrast to environments that switch repeatedly \citep{glaze2015,radillo2017}, and a declaration rule solely based on clicks and elapsed time. Within these assumptions, our sufficiency result is exact and prior-independent.

Our results describe the posterior at a single observation time $t$. A natural next step is to describe how it evolves as evidence arrives, so that the agent updates its belief continuously rather than recomputing it from the window at each moment. The posterior odds of \autoref{sec:poisson} indicate the form such dynamics should take. Without a deadline, the posterior odds can be updated recursively, so that a single running quantity summarizes the entire history \citep{shiryaev1963,poor2009}. With a deadline, change hypotheses enter as new candidates at age zero and expire once they reach age $w$, so the agent must track how its belief is distributed over the age of the change and not only its total. The expiring hypotheses remove weight from the accumulated odds at a rate that depends on the number of clicks in the most recent window, which acts as an evidence-dependent decay. Whether this decay is well approximated by the fixed, proportional leak of standard leaky-integration models, and under what conditions, is an open question that bears directly on how the timescales of evidence evaluation measured in behavior should be interpreted \citep{ganupuru2019,harun2020,booras2021,piet2018}.

A second direction is the optimality of the decision layer. Our closed-form evaluation gives the posterior probability on whether the change occurred, but not the policy by which the agent should make a declaration. Behaviorally, this policy is flexible: subjects in this task adjust their decision strategy after misses and false alarms \citep{johnson2017} and shift their decision bounds when cued to the direction of the change \citep{booras2021}. Given the response deadline, we conjecture that a simple constant-threshold declaration rule is suboptimal, in contrast to the classical quickest-detection setting, where a constant threshold is optimal \citep{wald1945,shiryaev1963,good1979,poor2009}. Instead, one must take into account how the posterior is distributed over the age of the change, which the posterior density $P(H_a \mid D)$ carries and the total probability $P(S_1 \mid D)$ discards. If most of the posterior evidence is recent, it may be worthwhile to wait for more evidence to accumulate. If strong evidence is approaching the end of the window, it may be more appropriate to make a declaration in order to avoid a potential miss. Window sufficiency greatly simplifies this problem, reducing each posterior calculation to a function of the latest window rather than the full task history. Beyond this, the model naturally extends to the addition of catch trials in which no change occurs, as in the experimental versions of the task \citep{johnson2017,ganupuru2019,harun2020}, a hard trial deadline $T$, or an unknown baseline rate. We leave these to future work.

\section*{Author Contributions}
T.D.H.\ conceived the project and contributed expertise from prior experimental work on the auditory change detection task in humans and rats. P.H.\ developed the theory, derived the results, prepared the figures, and wrote the manuscript. R.C.\ contributed to the theoretical development. R.C.\ and T.D.H.\ provided feedback throughout and reviewed and edited the manuscript. All authors approved the final version.

\section*{AI-Use Declaration}
The authors used Claude Opus 5.5 and Claude Opus 4.8 (Anthropic) to assist with editing portions of the text, formatting the manuscript in \LaTeX, and writing code for simulations and figures. The problem setup, theoretical results, proofs, and result interpretations remained in full control of the authors. All AI-assisted content was reviewed, revised, and verified by the authors, who take full responsibility for the manuscript.


\end{document}